\documentclass{sig-alternate}
\usepackage{subfigure}
\usepackage{graphicx}
\usepackage{url}
\usepackage{amsfonts}
\usepackage{array}
\usepackage{algorithm,algorithmic}
\usepackage{epsfig}
\usepackage{multirow}

\begin{document}

\title{Efficient Query Rewrite for Structured Web Queries}
\numberofauthors{1}
\author{
\alignauthor
Sreenivas Gollapudi, Samuel Ieong, Alexandros Ntoulas, Stelios Paparizos\\
    \affaddr{Microsoft Research, Silicon Valley}\\
    \email{\{sreenig, saieong, antoulas, steliosp\@microsoft.com}
}

\newcommand{\note}[1]{{\sc[[#1]]}}
\newtheorem{definition}{Definition}
\newtheorem{example}{Example}
\newtheorem{theorem}{Theorem}

\newcommand{\GreedyHistogram}{\textsc{Greedy-Rewrite}}
\newcommand{\DPHistogram}{\textsc{DP-Rewrite}}
\newcommand{\UnrestrictedHistogram}{\textsc{Attribute-Removal}}
\newcommand{\FuncDep}{\textsc{FD-Heuristic}}
\newcommand{\meandist}{\textsc{Mean-Dist}}
\newcommand{\RelaxedQuery}[1][m]{\ensuremath{\{a_1~:~v_1 \pm \delta_1, a_2:v_2 \pm \delta_2, \ldots, a_#1:v_#1 \pm \delta_#1\}}}
\newcommand{\Query}[1][m]{\ensuremath{\{a_1:v_1, a_2:v_2, \ldots, a_#1:v_#1\}}}
\newcommand{\WebQuery}[1]{`\texttt{#1}'}
\newcommand{\StructuredQuery}[1]{$\{$\emph{#1}$\}$}

\newcommand{\from}[2]{{\bf [{\sc from #1:} #2]}}

\pagenumbering{arabic}

\maketitle

\begin{abstract}



Web search engines and specialized online verticals are increasingly
incorporating results from structured data sources to answer semantically rich
user queries. For example, the query \WebQuery{Samsung 50 inch led tv} can be
answered using information from a table of television data. However, the users
are not domain experts and quite often enter values that do not match precisely
the underlying data. Samsung makes 46- or 55- inch led tvs, but not 50-inch
ones. So a literal execution of the above mentioned query will return zero
results. For optimal user experience, a search engine would prefer to return at
least a minimum number of results as close to the original query as possible.
Furthermore, due to typical fast retrieval speeds in web-search, a search
engine query execution is time-bound.

In this paper, we address these challenges by proposing algorithms that rewrite
the user query in a principled manner, surfacing at least the required number
of results while satisfying the low-latency constraint.  We formalize these
requirements and introduce a general formulation of the problem.  We show that
under a natural formulation, the problem is NP-Hard to solve optimally, and
present approximation algorithms that produce good rewrites.  We empirically
validate our algorithms on large-scale data obtained from a commercial search
engine's shopping vertical.

\end{abstract}

\section{Introduction}
\label{sec:introduction}


Web users are increasingly looking for information beyond the traditional
sources. This is manifested in search engines like {\tt google} and {\tt bing}
by the inclusion of answers beyond 10 page links and in the tremendous growth
of specialized search engines such as {\tt amazon}. Often the rich
experience is provided via the use of semantic information that comes
from (semi-)structured data sources in the form of tables, xml files or databases. For
example, structured data can be used to answer queries ranging such as
electronic goods (e.g. \WebQuery{50 inch samsung led tv}), fashion 
(e.g. \WebQuery{\$1600 prada handbags}), movie-showtimes listings (e.g.
\WebQuery{avatar showtimes near san francisco}), and weather
prediction (e.g. \WebQuery{weather in new york}).

A major challenge in using structured data to answer web queries is that users
often lack domain expertise and may pose queries that lead to very few or no
result due to unfamiliarity with the underlying data sources.  For example,
consider the query \WebQuery{50 inch samsung led tv}. There exists work in the
literature~\cite{LWA09, sarkas10} that can correctly classify and semantically
interpret the query to attribute-value pairs that correspond to underlying
structured attributes. So the query can be thought as ({\tt 50 inch}
$\Rightarrow$ {\tt display size}, {\tt Samsung} $\Rightarrow$ {\tt Brand}, {\tt
led tv} $\Rightarrow$ {\tt display type}). However, if the query is directly
evaluated as specified, there will be no results that can satisfy the
interpretation as Samsung does not make 50-inch LED TVs. On the other hand,
Samsung makes 46-inch and 55-inch LED TVs and 50-inch PLASMA TVs. Arguably, the
users would prefer to see such results that are close to their original query
instead of looking at an empty page with no results because they did not know
the appropriate precise values when typing the query.

The challenge is common to today's systems and not restricted to the
electronics domain but applies broadly to answering web queries in a variety of
domains including {\tt handbags} or {\tt shoes}, for example consider the query
{\tt \$1600 prada handbags}. One strategy for handling this challenge is to
\emph{rewrite the query} to broaden its coverage.  In the context of online
search, such rewrites include a variety of techniques such as query term
deletion, phrasal substitution, and mining of similar queries. In fact, the
query {\tt \$1600 prada handbags} does not return any products on {\tt
amazon} and is handled using term deletion as shown in
Figure~\ref{fig:amazon-example}. However this approach provides no quality
guarantees and does not take advantage of the rich meta-data information
available in structured data sources, thus producing results that leave a lot
to be desired to the user.

\begin{figure}[t]
\centering
\includegraphics[width=3.2in]{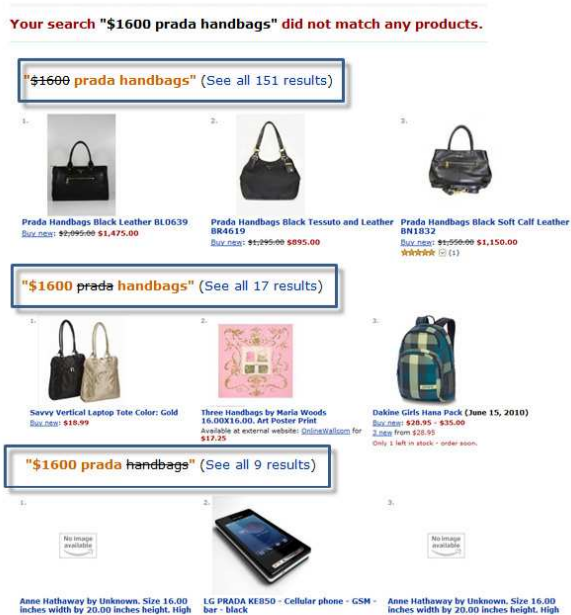}
\caption{The query {\tt \$1600 prada handbags} on {\tt amazon.com} is handled by dropping terms in the query successively and surfacing the results from each rewritten query separately}
\label{fig:amazon-example}
\end{figure}

We are interested in rewriting the queries through semantic term expansion. For
example, the above queries may be rewritten as \WebQuery{(46 to 52 inch)
(samsung or sony) (led or plasma) tv} and \WebQuery{(\$1400 to \$1800) (prada
or gucci) handbags} respectively. This query rewriting problem can be viewed as
a generalization of query rewriting through synonyms to increase recall, for
example, from \WebQuery{women shoes} to \WebQuery{(women or women's) (footwear
or shoes)}. Note that we are not interested in a set of static rewrite rules,
such as those used in synonym detection and stemming, but rather a query
rewrite algorithm that can understand the query intent and adapt accordingly.

The quality of the rewrites depends on two factors. First, the rewritten query
should preserve the meaning of the original query as closely as possible.  We
measure the fidelity of the rewrite by computing how far away the set of
results retrieved are to the original query, as estimated by user preferences
learned through click logs.  Second, the rewritten query should ensure there
are sufficiently many results returned to the user.  We measure the coverage of
the rewrite by counting how many times when a certain number of results is
requested, the expectation is met.

If efficiency had not been an issue, a candidate solution would have been to
expand the terms a little at a time, issue the rewritten query to the index,
and repeat as necessary until the minimum number of results requested is
retrieved.  This solution would not be applicable to web search, however, as
users expect results to be returned in under half a second, thus placing a strict
performance requirement on the query rewriting component. Hence, many search
engines place a restriction on the number of re-written queries (also called
query augmentations) that can be issued to the index as part of the original
query execution. To ensure this requirement is met, we require that the
techniques may only use precomputed statistics of the index but may not access
the index at run time, as index access contributes the lion's share of running
time.  Similarly, we also require that the techniques can take an input
parameter that limits the number of alternative rewrites they can examine.




In this paper, we formulate the above problem as time-bound query rewriting for
structured web queries.  As part of our contributions we formally describe an
optimization framework that takes as input a candidate query $q$, a desired
number of results $k$, and a parameter $T$ that governs how many rewrites can
be considered, and produces a rewritten query that aims to retrieve at least
$k$ results and that the results match the original query well.  We show that
finding the optimal solution to this problem is NP-Hard.  We introduce a greedy
algorithm and a dynamic programming solution that rewrite the query in a
principled and controlled fashion.  We also study the effect of functional
dependencies in the data and how they affect query rewrite. We evaluate the
proposed solution using real queries from a commercial search engine's shopping
vertical against a prototype commerce search engine.


The rest of the paper is organized as follows. In Section 2, we discuss related work.  In Section 3, we describe our
model and assumptions about structured web search, and formulate the problem of
predicate relaxation.  To meet the performance requirement, one needs to
pre-compute statistics on the database to be used at runtime.  In Section 4, we
describe two kinds of statistics---histograms and functional dependencies---and
give two heuristics for using these statistics to perform fast predicate
relaxation.  In Section 5, we report our experimental evaluation of these
heuristics conducted over data from a commercial search engine's vertical. We summarize and conclude
in Section 6.

\section{Related Work}
\label{sec:related}

Structured data is abundant on the web, and there have been studies on how to retrieve them in a manner suitable to web search~\cite{bergman2001deep, CHWWZ08}.  There is also work on how to retrieve and rank information from structured data~\cite{hristidis03efficient,CSLRV04, BCS09,KXC09, LYMC06}. When answering web queries over structured data, however, direct application of textual similarity may produce low quality results due to possible misinterpretations of data types. For example, a database might store the television diagonal as the string `50 inches' while users may type \WebQuery{50"}. To this end, recent work has studied how to analyze keyword queries as typed in a web search box and interpret them as structured queries~\cite{LWA09, sarkas10}.  These past works form the basic components over which we build our system for answering web queries using structured data.


Rewriting user queries to broaden coverage is a common technique employed by all search engines. For example, search engines routinely make spelling corrections to queries when retrieving results.  In the context of search over structured data sources, textual similarity approaches that treat the query as a bag of words will generally perform poorly.  In the example query given in the Introduction, there is no textual relaxation between \emph{Samsung} and \emph{Sony}, and little can be done for generating term expansions or substitutions for the diagonal size in a controlled manner.  Past approaches based on log mining~\cite{AGC08Simrank++,JF03Word-Deletion,JRMG07Query-Substitutions} may be able to discover relationship between terms that do not exhibit textual similarity, but they do not address how such knowledge can be exploited in conjunction with statistics of the documents to come up with good rewrites of the queries that preserve fidelity and ensure coverage.

Fontoura \textit{et al.} proposed a method to relax text queries using taxonomies~\cite{FJKOTV08}.  Their approach can also be viewed as rewriting queries taking advantage of a taxonomy created by experts, and thus solves a similar problem to ours.  However, creating a good taxonomy requires significant domain knowledge, and is an expensive process.  In our application domain, we do not have such a taxonomy available, and hence the work is not directly comparable.

These has been work in the database community that investigate the problem of keyword search over structured data~\cite{CHWWZ08,CSLRV04,hristidis03efficient,KXC09,LYMC06}.  They assume that the expansion of the keywords is handled through some probabilistic methods or captured in the ranking function, and focus on performance issues.  Like these work, we are concerned about performance issues, and capture the requirements by explicitly specifying them in our framework; our work is different in that it allows a more controlled behavior in the rewriting that provides quality guarantees.

Finally, given a query and a distance function, one can think of the problem we are trying to solve as a nearest neighbor problem. Nearest neighbor problems have been studied in the past, for example~\cite{roussopoulos-sigmod95}. More recently there are even k-nearest neighbor considerations, like~\cite{beecks-dbrank10, zhang-icde10, 293348, 335428, SK98, XMMR08}, that are applicable in the setting of searching over a database. Similar to the k-nearest neighbors, but from a join relaxation problem in databases is the work described in~\cite{KLTV06}. We find such work very valuable in relaxing the user query and finding good quality results within a reasonable distance around what was specified in the query. Although useful, the techniques described have a fundamental difference with our work. In the web search over structured data setting we need both quality guarantees with regards relaxation but at the same time we have strict performance guarantees requiring an upper time bound.  Further, these approaches admit relaxations of numeric attributes only and extensions to categorical attributes are non-trivial.  In contrast, our approaches come with two advantages - 1) they are very simple to implement; and 2) support distance functions on both categorical and numeric attributes.  In fact, we will precisely use one such distance function in our experiments and show that our algorithms perform well in practice.


\section{Problem Formulation}
\label{sec:background}

We first describe a model of structured web queries, and assumptions on how they are parsed, and how
items are evaluated with respect to the parsed queries.  We then formally define
the problem of time bound query rewrites.

\subsection{Model}

Given a keyword web query, we assume the existence of a semantic parser that
identifies the attributes requested in the query and extracts their associated
desired values, based on past work such as~\cite{LWA09, sarkas10}. For example,
the query \WebQuery{50 inch samsung led tv} is parsed as a \emph{structured query}
\StructuredQuery{table:TV, brand:Samsung, type:LED, diagonal:50}.  Denote a
generic parsed query by its attribute-value pairs, $q = \Query$.  Denote the value
of attribute $a_i$ in query $q$ by $q_{a_i}$.  For our example query, $q_{\text{brand}} = \text{Samsung}$.  Consistent with the interpretation of web
queries as conjunctions of keywords, we interpret the structured query under
the AND-semantics as well. For the rest of the paper we assume that structured queries are given to us in the form of
attribute-value pairs.  In practice, not all terms in a query will be understood by the parser.  The terms that are not understood are treated as keywords used by the ranking function as additional signals.

Let $P$ be a database of items, from which we retrieve results to serve the query.  For each item $p \in P$, we represent it as a set of attribute-value pairs $\{a_1 : w_1, a_2 : w_2, \ldots, a_n : w_n\}$, and the value of attribute $a_i$ by $p_{a_i}$.  We assume that the semantic parser will only identify attributes for which we have data, hence the query specifies the values of a subset of these $n$ attributes.  Henceforth, when a query $\Query$ is given, we only focus on the $m$ attributes mentioned.  We give an example table of the database for TVs in Table~\ref{tab:database}, which we use throughout the paper for illustration.  The size of the database will be significantly larger in practice.


As discussed in the Introduction, users may lack domain expertise and may be unfamiliar with the attribute specification of the underlying structured data.  Consider the sample query \StructuredQuery{table:TV, brand:Samsung, type:LED, diagonal:50}. For the database table in Table~\ref{tab:database}, there is no TV that matches all the requested attribute values.  Nonetheless, it is desirable that a search engine should return results that are close to the query, for example, Samsung LED TVs of 46 inches or 55 inches, or Sharp LED TVs of 52 inches.  It would be less desirable, however, if the search engine returns a Samsung LED TV of 32 inches, since that TV is much smaller than requested, or a Sony CRT TV of 50 inches, since the type of TV is significantly different than requested.

\begin{table}
\centering
\begin{tabular}{|l|l|l|l|l|}
\hline
\textbf{Brand} ($b$) & \textbf{Model} ($m$) & \textbf{Type} ($t$) & \textbf{Diagonal} ($d$) & $\ldots$ \\
\hline
Samsung	& UN46B6000	   & LED & 46 & $\ldots$\\
Samsung	& UN55B7000	   & LED & 55 & $\ldots$\\
Samsung	& UN32B6000    & LED & 32 & $\ldots$\\
Samsung	& LN55B630     & LCD & 55 & $\ldots$\\
Samsung	& PN46A550     & Plasma & 46 &$\ldots$\\
Sony    & KDL-52XBR9   & LCD & 52 & $\ldots$\\
Sony    & KDL-46EX700  & LCD & 46 & $\ldots$\\
Sony    & KD-50FS170   & CRT & 50 & $\ldots$\\
Sharp   & LC-52D85UN   & LED & 52 & $\ldots$\\
Sharp   & LC-52LE700UN & LCD & 52 & $\ldots$\\
\hline
\end{tabular}
\caption{Example database for TVs.} \label{tab:database}
\end{table}

\begin{table}
\centering
\begin{tabular}{|l|l|l|l|}
\hline
\textbf{Attribute} $i$ & $v$ & $w$ & $d_i(v, w)$ \\
\hline
\multirow{3}{*}{Brand ($b$)}
 & Samsung & Sony & $0.2$ \\
 & Samsung & Sharp & $0.3$ \\
 & $\ldots$ & $\ldots$ & $\ldots$ \\
\hline
\multirow{4}{*}{Type ($t$)}
 & LED & LCD & $0.1$ \\
 & LED & Plasma & $0.5$ \\
 & LED & CRT & $1.0$ \\
 & $\ldots$ & $\ldots$ & $\ldots$ \\
\hline
\multirow{5}{*}{Diagonal ($d$)}
 & 50 & 32 & $0.8$ \\
 & 50 & 46 & $0.3$ \\
 & 50 & 52 & $0.1$ \\
 & 50 & 55 & $0.4$ \\
 & $\ldots$ & $\ldots$ & $\ldots$ \\
\hline
\end{tabular}
\caption{Example distance function for TVs.} \label{tab:distance}
\end{table}

To make the discussion formal, denote the domain of attribute $a_i$ by $A_i$.
Let the function $d_i : A_i \times A_i \rightarrow [0, 1]$, $d_i(v, w)$ measures
the distance of attribute value $w$ from attribute value $v$.  When $d_i(v, w)$
is small, it means that attribute value $w$ is similar to attribute value $v$.
We give an example distance function for TVs in Table~\ref{tab:distance}.  We
note that our solution does not depend on assumptions such as symmetry or
triangle-inequality about the distance function.

An aggregate distance function $ad : P \times Q \rightarrow \mathbb{R}$, $ad(p, q)$ measures how well item $p$ matches query $q$.  When $ad(p, q)$ is small, it means item $p$ matches the query $q$ well.  We assume that the function depends only on the attribute values of the item and the query.  We next define a basic yet fundamental property of aggregate distance functions that we assume throughout the paper.

\begin{definition}[Monotonicity]
  An aggregate distance function, $ad(\cdot)$, is \emph{monotonic} if for any query $q = \Query$, any two items $p^1$ and $p^2$, if
  \[
    \forall 1 \leq i \leq m, \qquad d_i(v_i, p^1_{a_i}) \geq d_i(v_i, p^2_{a_i}) \enspace,
  \]
  then $ad(p^1, q) \geq ad(p^2, q)$.
\end{definition}

Monotonicity ensures that an item closer to the query in each of the attributes will also be closer to the query in aggregate distance.  This is a natural property that should be satisfied when the attribute distances determine how well an item matches a query.  Example aggregate distance functions that satisfy monotonicity includes weighted sums of the attribute distances, and $\ell_p$-norms that treat attribute distances as vectors in $m$-dimensional space.

It is possible that a search engine may choose a ranking function that does not satisfy monotonicity.  This happens when the ranking function takes into account additional sources of signals such as click activities in deciding how well an item matches a query.  This is outside of the scope of our problem formulation.

\subsection{Query Rewrite Formulation}

Given a query $q$, our ultimate goal is to find the top-$k$ items that match the query
within a fixed time window. In order to find the top-$k$ items, we must
first be able to select at least $k$ items from the database.  This may not be
possible when there are less than $k$ items that match all the desired attribute values.  The
focus of our work is on how to rewrite the query in a principled way so as to
ensure sufficiently many items are returned, keeping fidelity to the original query,
while respecting the time constraints imposed on the individual components of a search engine.

For a given attribute-value pair $a_i : v$ and value $\delta \in [0, 1]$, let $B_i(v, \delta)$ be the set of attribute values that is $\delta$-close to $v$, i.e.,
\begin{equation}
  B_i(v, \delta) = \{w \in A_i | d_i(v, w) \leq \delta\} \enspace.
\end{equation}
In other words, $B_i(v, 0)$ are the set of attribute values that are equivalent to $v$, whereas $B_i(v, 1)$ are the set of all attribute values.  For example, for the distance function in Table~\ref{tab:distance}, $B_d($`Samsung'$, 0.2)$ $= \{$`Samsung'$, $`Sony'$\}$.

Denote a \emph{relaxed query} $q$ by $\RelaxedQuery$.  A database item $p$ matches $q$ if and only if
\[
  \forall 1 \leq i \leq m, \qquad p_{a_i} \in B_i(v_i, \delta_i) \enspace.
\]



At a high level, the query rewrite for structured web query problem is to take an input query and find a relaxed query that will result in at least $k$ matches in the database.
If time had not been an issue, a simple solution would be to iteratively make small relaxation to the query, issue it to the database to find out the number of matches, and repeat until we have found $k$ results. However, due to the performance requirement imposed by web search, this approach is infeasible as database access is costly.  Indeed, an algorithm may only be able to carry out a small amount of computations within the time envelope.

To capture these limitations, we propose to bound the time of any solution by the \emph{number of different relaxed queries} it considers and include this as an explicit parameter to the problem specification.  To ensure that this meaningfully reflects the performance requirement and is helpful in differentiating among solutions, we require that the amount of time it takes to evaluate each relaxed query to be constant.  Note that different forms of evaluating the relaxation will lead to different classes of problems.  For example, evaluation via issuing the relaxed query to a database will constitute a different class of problems from evaluation via approximation by database statistics.  Indeed, in this paper, we focus on the latter form of evaluation, which we made clear in Section~\ref{sec:relaxation}.  We model the performance requirement using this abstract bound in place of an actual time parameter as the actual amount of time needed varies across systems and is dependent on the quality of the implementation.

We now give a formal definition of the problem.

\begin{definition}[Time Bound Query Rewrite]
Given:
\begin{itemize}
\item A query $q = \Query$;
\item A database of items $P = \{p^1, p^2, \ldots, p^n\}$;
\item The minimum number of items to return, $k$;
\item The maximum number of relaxed queries considered, $T$.
\end{itemize}


\noindent Find a relaxed query $q' = \RelaxedQuery$ with at most $T$ relaxed queries considered, such that the number of items that match the query $q'$, $S \subseteq P$, is at least $k$, and that the average aggregate distance among all items in $S$ from query $q$,
\begin{equation}
  ad(S, q) = \frac{1}{|S|} \sum_{p \in S} ad(p, q) \enspace, \label{eq:aggregate-distance}
\end{equation}
is minimized.
\end{definition}

\section{Statistics and Heuristics}
\label{sec:relaxation}
\newcommand{\Est}{\textsc{Est}}
\newcommand{\PredRelaxHistogram}{\textsc{Query-Rewrite-Histograms}}

To enable fast evaluation of candidate relaxed queries, one can precompute statistics on the database, and estimate the number of matches using these statistics.  We consider two statistics---histograms of attribute values and attribute dependencies estimated as conditional probability distributions---which are commonly computed in databases, and formulate a version of time bound query rewrite problem.  We then present two heuristics, one based on a greedy approach, and another based on dynamic programming, and discuss trade-offs between the two approaches.

\subsection{Statistics}

\subsubsection{Histograms}

One of the most important statistics of an attribute is the distribution of its values, termed the \emph{histogram}.  Histograms can help to provide estimate of the number of potential matches to a query without direct database access.

Formally, let the histogram of attribute $a_i$ be $h_i$, and that for a set of attribute values $V \subseteq A_i$, $h_i(V)$ returns the number of items that have the corresponding attribute value.  For example, the histogram for the brand attribute in our example database would be
\[
  h_b(\text{`Samsung'}) = 5 \qquad h_b(\text{`Sony'}) = 3 \qquad h_b(\text{`Sharp'}) = 2 \enspace.
\]

If one assumes that the attributes in the query are independent, one can estimate the number of matches to a relaxed query as follows.  For a query $q = \RelaxedQuery$, the estimated number of matches equals
\begin{equation} \label{eq:estimate}
  \Est(q) = |P| \prod_{i=1}^m \frac{h_i(B_i(v_i, \delta_i))}{|P|} \enspace.
\end{equation}
As an example, for $q = \{$brand$=$Samsung $\pm 0.2$, type$=$LED $\pm 0.2$, diagonal$=$50 $\pm 0.3$ $\}$,
\begin{align*}
  \Est(q) = &10 \Bigl(\frac{h_b(B_b(\text{`Samsung'}, 0.2))}{10}\Bigr) \Bigl(\frac{h_t(B_t(\text{`LED'}, 0.2))}{10}\Bigr)\\ &\quad\Bigl(\frac{h_d(B_d(50, 0.3))}{10}\Bigr) \\
  =& 10 (0.8) (0.8) (0.7) = 4.48
\end{align*}

When attributes are dependent, the estimate could be misleading.  Functional dependencies may help to improve the estimate.

The maximum aggregate distances of the set of selected items to a query cannot be determined by the histograms alone.  Therefore, one cannot directly optimize objective~\eqref{eq:aggregate-distance}.  Instead, we focus on bounding the aggregate distance by controlling the total amount of relaxation, and define the problem of query rewrite using histograms as follows.

\begin{definition}[\PredRelaxHistogram]
Given
\begin{itemize}
\item A query $q = \Query$;
\item Database size $|P|$;
\item Histograms $h_i$ for each attribute $a_i$;
\item The minimum number of items to return, $k$;
\item The maximum number of relaxed queries considered, $T$.
\end{itemize}


\noindent Find a relaxed query $q' = \RelaxedQuery$ with at most $T$ relaxed queries considered, such that $\Est(q')$ is at least $k$, and that the total amount of relaxation,
\begin{equation}
  tr(q') = \sum_{i=1}^m \delta_i \enspace, \label{eq:sum-delta}
\end{equation}
is minimized.
\end{definition}

Later in this section, we show that this problem is hard (even in the absence of
a limit on the number of relaxed queries considered), and propose heuristics for
solving this problem.

\subsubsection{Attribute Dependencies}

Suppose a query specifies both a brand and a model.  Consider the example database in Table~\ref{tab:database}.  If only one of the two attributes is relaxed, there will be no additional matches for the relaxed query.  Yet the estimate using Equation~(\ref{eq:estimate}), based on the assumption that attributes are independent, would erroneously estimate that the number of matches will increase after the relaxation.  To address this problem, one has to account for attribute dependencies in the database.

We start by precomputing the conditional probabilities $P(a_i = v_i | a_j =
v_j)$ for all pairs of attributes $a_i$ and $a_j$ in the database.  For query
$q = \Query$, if $P(a_i = v_i | a_j = v_j)$ is higher than some threshold, this
indicates that the attributes are dependent, and we propose to drop either
$a_i$ or $a_j$ from the query.  We believe there are good arguments for either
approach to perform better; it depends on whether we have a better distance
function for attribute $a_i$ or $a_j$.  To test the effect of attribute
dependencies, we evaluated both possible directions in our experiments.

After this preprocessing step, we apply the same techniques for
\PredRelaxHistogram\ on the modified instance.  It may be possible to use the
conditional probabilities in a finer-grained manner to further improve the
query rewriting process; we leave that for future work.

\subsection{Hardness of \PredRelaxHistogram}

The problem of \PredRelaxHistogram\ is closely related to knapsack problems, and is hard to solve optimally.

\begin{theorem}
  \PredRelaxHistogram\ is NP-hard, even in the absence of a bound on the maximum number of relaxations considered.
\end{theorem}

\begin{proof}
  We reduce \textsc{Subset-Product}, an NP-hard problem, to a decision version of \PredRelaxHistogram($c$) where we ask if there exists a relaxation for which $tr(q') \leq c$.

  The \textsc{Subset-Product} (SP14, \cite{garey-johnson79}) is as follows.  Given a finite set $A$, a size $s(a) \in \mathbb{Z}^+$ for each $a \in A$, and a positive integer $B$,  is there a subset $A' \subseteq A$ where $\prod_{a \in A'} s(a) = B$.

  We create an instance of \PredRelaxHistogram($c$) as follows.  We map each element in the finite set $A$ to an attribute.  Create a query $q = \Query$, where $m = |A|$, and a database of size $|D| > \max_{a \in A} s(a)$, the latter serves as a normalization constant for our problem.  For each item $a \in A$ with size $s(a)$, create a histogram for attribute $a$ with
  \[
    h_a(B_a(v_a, t)) =
      \begin{cases}
        1 & \text{for } 0 \leq t < \log s(a) \\
        s(a) & \text{for } t \geq \log s(a)
      \end{cases} \enspace.
  \]
  Set $k = B / |D|^{m-1}$, and the decision parameter $c = \log B$.

  The instance of \PredRelaxHistogram($c$) evaluates to YES if and only if there exists a relaxed query $q' = \RelaxedQuery$ satisfying
  \begin{align*}
    \Est(q) &= \frac{\prod_{i : \delta_i = \log s(a)} s(a) }{|D|^{m-1}} \geq \frac{B}{|D|^{m-1}} \\
    tr(q) &= \sum_{i : \delta_i = \log s(a)} \log s(a) \leq c = \log B \enspace,
  \end{align*}
  which is possible only if $\prod_{i : \delta_i = \log s(a)} s(a) = B$.

  One loose end remains is that the exact values $\log s(a)$ and $\log B$ are not representable in finite number of digits.  We need to show that the reduction continues to hold after rounding these input to some precision $\epsilon$, and that $\log (1 / \epsilon)$ is polynomial in the size of the \textsc{Subset-Product} instance.  When $\log s(a)$ and $\log B$ can have at most an error of $\epsilon$, for $tr(q)$ to be smaller than $\log B$ but not $\log (B+1)$, we need
  \begin{align*}
    (\log B + \epsilon) + n \epsilon &< (\log (B+1) - \epsilon)\\
    (n+2) \epsilon &< \log((B+1) / B) \leq 1/B \\
    \epsilon &< 1 / ((n+2) B) \enspace,
  \end{align*}
  or $\log (1 / \epsilon) = O(\log nB)$, as desired.
\end{proof}

Therefore, in order to solve the problem, we rely on heuristical approaches for solving the problem.
\eject
\subsection{Algorithms for Query Rewrite}

\subsubsection{\GreedyHistogram}

A general template for solving \PredRelaxHistogram\ is to (1) select an attribute based on some criteria, (2) relax it by a small amount $\epsilon$ to get relaxed query $q$, (3) compute the estimate $\Est(q)$, and (4) repeat as long as $\Est(q) < k$.  Different choices of selection criteria give rise to different heuristics.

In \GreedyHistogram, we select an attribute to relax based on how constraining the attribute is.  Formally, for a relaxed query $q = \RelaxedQuery$, we pick the most constraining attribute, $a_i$ where
\[
  h_i(B_i(v_i, \delta_i))
\]
is the smallest to relax.

As an example, consider again the query $q = $ \StructuredQuery{table:TV, brand:Samsung, type:LED, diagonal:50}, the target number of results be $3$, and the maximum number of relaxed queries considered be $T = 10$.  Let the step size $\epsilon = 0.1$ for all attributes.  \GreedyHistogram\ will proceed as in Table~\ref{tab:greedy-example}.  At termination, it returns the relaxed query \StructuredQuery{table:TV, brand: Samsung $\pm 0.2$, type:LED $\pm 0.1$, diagonal:50$\pm 0.3$}, which yields $3$ results in our example database.

\begin{table}
\centering
\begin{tabular}{|c|c|c|c|c|c|c|c|}
\hline
\textbf{Step} & $\delta_b$ & $\delta_t$ & $\delta_d$ & $h_b(\cdot)$ & $h_t(\cdot)$ & $h_d(\cdot)$ & $\Est$ \\
\hline
0 & 0.0 & 0.0 & 0.0 & 5 & 4 & 1 & 0.20 \\
1 & 0.0 & 0.0 & 0.1 & 5 & 4 & 4 & 0.80 \\
2 & 0.0 & 0.1 & 0.1 & 5 & 8 & 4 & 1.60 \\
3 & 0.0 & 0.1 & 0.2 & 5 & 8 & 4 & 1.60 \\
4 & 0.0 & 0.1 & 0.3 & 5 & 8 & 7 & 2.80 \\
5 & 0.1 & 0.1 & 0.3 & 5 & 8 & 7 & 2.80 \\
6 & 0.2 & 0.1 & 0.3 & 8 & 8 & 7 & \textbf{4.48} \\
\hline
\end{tabular}
\caption{\GreedyHistogram\ with $\epsilon = 0.1$.} \label{tab:greedy-example}
\end{table}

If at the end of having evaluated $T$ relaxed queries and none is found to have an estimated number of matches of at least $k$, the last relaxed query (i.e., the one with the largest amount of relaxation) is returned.

\subsubsection{\DPHistogram}

Drawing on ideas similar to the dynamic program for knapsack-style problems, we also consider a dynamic programming heuristic \DPHistogram.  For a query $q = \Query$, let
\begin{quote}
  $F(j, d) = $ Maximum fraction of products satisfying the relaxed query on attributes $a_1, \ldots, a_j$ with total relaxation $\sum_{i=1}^j \delta_i \leq d$.
\end{quote}
Let $\epsilon$ be a parameter to the heuristic that determines the step size, i.e., by what increment we increase the relaxation of an attribute.  For each cell in $F(\cdot,\cdot)$, we need to consider one new relaxation.  Therefore, for a given maximum number of relaxations $T$, we can consider only $\rho = \lfloor\frac{T}{m}\rfloor$ different values for each attribute.  Hence, we compute $F(j, d)$ using dynamic programming as described in Algorithm~\ref{alg:dp}.

\begin{algorithm}
\begin{algorithmic}
\FOR{$d \leftarrow 0, \epsilon, 2\epsilon, \ldots, \min(\rho\epsilon, 1)$}
  \STATE{$F(1, d) \leftarrow \displaystyle \frac{h_1(B_1(v_1, d))}{|P|}$}
\ENDFOR
\FOR{$j \leftarrow 2 $ to $m$}
  \FOR{$d \leftarrow 0, \epsilon, 2\epsilon, \ldots, \min(\rho\epsilon, j)$}
    \STATE{$F(j, d) \leftarrow \displaystyle\max_{d'=0,\epsilon,\ldots,\min(d,1)} \Bigl(\frac{h_j(B_j(v_j, d'))}{|P|} F(j - 1, d - d')\Bigr)$}
  \ENDFOR
\ENDFOR
\end{algorithmic}
\caption{Dynamic program for Query Rewrite Using Histograms, with step size $\epsilon$ and $\rho = \lfloor\frac{T}{m}\rfloor$.} \label{alg:dp}
\end{algorithm}

The optimal solution is given by $\min_{d'} F(m, d')$ for which the value is at least $\frac{k}{|D|}$.  The amount of relaxation for each attribute can be kept track of by an auxiliary table.

Consider again the query $q = $ \StructuredQuery{table:TV, brand:Samsung, type:LED, diagonal:50}, the target number of results be $3$, and the maximum number of relaxations considered be $T = 15$.  Let the step size $\epsilon = 0.1$, a sample execution of \DPHistogram\ is illustrated in Table~\ref{tab:dynamic-program}.  At termination, it returns the relaxed query \StructuredQuery{table:TV, brand:Samsung $\pm 0.3$, type:LED $\pm 0.1$, diagonal:50$\pm 0.1$}, which yields $3$ results in our example database.  Note that, however, if $T = 10$, then $\rho = 3$, and hence the algorithm will only be able to evaluate up to $F(3, 0.3)$, and will fail to find a relaxation.

\begin{table}
\centering
\begin{tabular}{|c|c|c|c|}
\hline
 & \textbf{Attr 1} ($b$) & \textbf{Attr 2} ($t$) & \textbf{Attr 3} ($d$) \\
\hline
$d$ & $F(1, d)$ & $F(2, d)$ & $F(3, d)$ \\
\hline
0.0 & 0.50 & 0.50 * 0.40 = 0.20 & 0.20 * 0.10 = 0.020 \\
0.1 & 0.50 & 0.50 * 0.80 = 0.40 & 0.20 * 0.40 = 0.080 \\
0.2 & 0.80 & 0.50 * 0.80 = 0.40 & 0.40 * 0.40 = 0.160 \\
0.3 & 1.00 & 0.80 * 0.80 = 0.64 & 0.40 * 0.40 = 0.160 \\
0.4 & 1.00 & 1.00 * 0.80 = 0.80 & 0.64 * 0.40 = 0.256 \\
0.5 & 1.00 & 1.00 * 0.80 = 0.80 & 0.80 * 0.40 = \textbf{0.320} \\
\hline
\end{tabular}
\caption{\DPHistogram\ with $\epsilon = 0.1$, $\rho = 15 / 3 = 5$, and $k = 3$, i.e., $\frac{k}{|P|} = 0.3$.} \label{tab:dynamic-program}
\end{table}

Similar to \GreedyHistogram, if no relaxed query with an estimated number of matches of at least $k$ is found at the end of having evaluated $T$ relaxed queries, the relaxed query with the largest amount of relaxation is returned.



\subsubsection{Trade-off Between the Heuristics}

There is a trade-off between the two heuristics described.  On the one hand, for any fixed $\epsilon$, if the maximum number of relaxed queries allowed is large, \DPHistogram\ is guaranteed to find a relaxed query $q$ with $tr(q)$ no larger than the one found by \GreedyHistogram.\footnote{Note that this does not guarantee the results returned by \DPHistogram\ is necessarily better than ones returned by \GreedyHistogram\ when measured in the objective of Equation~\eqref{eq:aggregate-distance}, since aggregate distance and total relaxation is not equivalent.}  However, when the number of relaxed queries allowed is small, \DPHistogram\ will be able to investigate solutions of only small total amount of relaxation, and fails to find a solution when \GreedyHistogram\ may succeed.   We explore this trade-off more fully in the experiments.

\section{Experimental Evaluation}
\label{sec:experiments}

In this section, we study the behavior and performance of our algorithms on
effectively rewriting real user queries.
\eject
\subsection{Experimental Setup}

For our experimental evaluation we built a prototype search engine and we
populated it with real data from the shopping vertical of a commercial search
engine. To this end, we downloaded the detailed descriptions for about 5
million products related to 73 categories about electronics (such as
Televisions, Equalizers, GPS Receivers, etc.) from~\cite{msn-shopping}. Each
product is provided in structured form with its attributes clearly specified
like~\cite{msn-shopping-product}. We indexed the product details and computed
the histograms for the attributes as described in Section~\ref{sec:relaxation}.

As our query set we used a random sample of one thousand queries of a
major commercial search engine's log that were provided to us.  We
selected the queries that were directed to the categories described
above and for which we extracted attribute value pairs to form the
corresponding structured queries. We use well-known techniques~\cite{LWA09, sarkas10} to extract the attribute information from the queries. The categorization and translation to structured queries was verified manually to be correct.

We ran all thousand queries through our system and we selected the ones that
triggered a query rewrite because they return too few
(less than $k$) results. Out of the thousand queries, 343 would benefit
from query rewrites. Since the queries were a random sample of queries
targeted towards the structured data that we have available, on
average approximately 34\% of such queries could potential benefit. In the remainder of this section, we use these
343 queries as our query set to evaluate in depth our techniques.

\subsubsection{Comparison Method}

As observed in the Introduction, for queries that trigger very few results, \texttt{amazon.com} rewrites the query by dropping words from the query.  To take advantage of the semantics parser, instead of dropping words from the query, we implemented a version that removes attributes from the structured interpretation of the query.  The attribute to remove is selected based on which attribute is the most constraining.  We compare our method to this approach which we termed \UnrestrictedHistogram.  We present its performance in Section~\ref{ssec:num-steps}.  Note that there is no parameter to tune for this algorithm.

\subsubsection{Distance Function} \label{sssec:distance-function}

Within our prototype search engine, we also implemented a distance function to
be used for ranking and evaluating our results after query rewrite. As our
aggregate distance function $ad(\cdot)$ we considered the average distance of
the query to the items in our data set. More specifically, for a given query
$q=\Query$ and an item $p$, $ad(p, q) = \frac{1}{m}\sum_{i} d_i(v_i, p_{a_i})$,
where $d_i(v_i, p_{a_i})$ is the individual distance between the $q$ and $p$
for attribute $a_i$.

One natural definition of distance (or similarity) between attribute values is based on the notion of {\em substitutability}, i.e., the likelihood of a user substituting her desired attribute value $v$ (specified in the query) by eventually choosing a product with a different attribute value $v'$.  For example, a user looking for a \emph{nikon} digital camera is much more likely to substitute the brand for another well-recognized brand such as \emph{canon} rather than an obscure one like \emph{yashica}.  Thus, the distance between {\em nikon} and {\em canon} is expected to be smaller than that between {\em nikon} and {\em yashica}.  Similar intuition holds for a numerical attribute as well.  Consider a user buying a $32 inch$ lcd tv.  She is more likely to eventually buy a $36 inch$ than a $60 inch$ lcd tv.

In our implementation, we define $d_i$ as the normalized distance of the two attribute values when they are numeric, i.e., $d_i(v_i, p_{a_i}) = \min(1.0, \frac{|v_i - p_{a_i}|}{|v_i|})$. For categorical attributes, we compute this distance measure using a methodology similar to the one described in \cite{PG11} based on browsed trails originating from search engines.  As these distances are based on search logs, certain attribute values appear very rarely, leading to no estimate for certain pairs of attribute values.  For example, for the attribute \emph{model}, distances between pairs of \emph{model numbers} could be missing.  In such cases, we take the conservative position that the missing distances to be the maximum possible distance of $1$.


For our performance metric \meandist, we will use the mean distance (as captured by $ad(\cdot)$) over all items in our result set, i.e. we will use Equation~\eqref{eq:aggregate-distance}.  To penalize for the cases where the algorithm fails to find at least $k$ results, which could happen due to poor estimates that overestimates the number of matches of a relaxed query, or an algorithm having attempted $T$ different relaxed queries, we treat any shortfall as having retrieved documents that are at a maximum possible distance of $1$ from the query.  Under this \emph{penalty}, an algorithm that finds a relaxed query that obtains at least $k$ results will do better than one that does not.

Finally, for the experiments presented in this section we set the number of returned results $k=10$.

\subsection{Varying the Step Size}

We start our experimental evaluation by studying the effect of the step size
$\epsilon$ in the performance of our query rewrite algorithms.  Both
\GreedyHistogram\ and \DPHistogram\ use a parameter $\epsilon$ that determines
the amount of relaxation of an attribute at a step of the algorithm.
Intuitively, for small $\epsilon$, we are making smaller, more careful steps
when relaxing so we expect that the furthest item will be quite close to the
$k^{th}$ item. On the other hand, if $\epsilon$ is large, we are relaxing more
aggressively and we may identify significantly more than $k$, and thus our
performance metric may be worse.

To study this effect in more detail, we evaluated our algorithms over our data
and we plot the graphs shown in Figure~\ref{fig:greedy-perf-over-epsilon} for
\GreedyHistogram\ and in Figure~\ref{fig:dp-perf-over-epsilon} for
\DPHistogram. The results for \UnrestrictedHistogram\ is not affected by the step size $\epsilon$ or the number of steps $T$.  The data is shown in Figure~\ref{fig:algs-perf-over-timesteps} and is not shown in Figures~\ref{fig:greedy-perf-over-epsilon} and \ref{fig:dp-perf-over-epsilon} for presentation clarity.

The algorithms were allowed upto a total of 20 steps, which ensured that they
would consider rewrites that would return at least $k$=10 results. The horizontal axis shows increasing values of $\epsilon$ and the vertical axis shows the average \meandist\ at a given $\epsilon$ value. Lower values in the vertical
axis indicate better performance.

In the case of \GreedyHistogram, we observe that increasing step sizes
lead to a larger value under our performance metric, i.e., worse results. As our algorithms become
more aggressive (increasing $\epsilon$) they allow for the result set to grow
much larger than $k$ and thus \meandist\ increases. Of course, smaller
$\epsilon$ values imply better performance but at the cost of requiring more
steps until completion.

The picture for \DPHistogram\ is more complicated. When the number of steps is very few, it faces a trade-off in choosing the step size.  When the step size is small, \DPHistogram\ fails to find rewrites that retrieve at least $k$ results, leading to poor performance as it is penalized for the shortfall; when the step size is large, \DPHistogram\ finds rewrites that obtains at least $k$ results, but now with large total amount of relaxation across all attributs.  Hence, we see a U-shaped curve for small number of steps.  When the number of steps is large, the performance of \DPHistogram\ is closer to monotonically increasing in step sizes, as relaxations of at least $k$ results are found for any step sizes, and hence smaller step sizes lead to better performance.  Indeed, we see that the three curves for number of steps $= 12, 16, 20$ overlaps one another, indicating that the same rewrite is found.  The small dip from $\epsilon=0.2$ to $\epsilon=0.3$ is due to a couple of queries where the best relaxation is by rewriting an attribute to include values that are $0.3$ (and $0.9$) away, whence for $\epsilon=0.2$ these attributes have to include values that are $0.4$ (and $1.0$) away.

In both cases, we found that $\epsilon=0.1$ gives a reasonable performance for our practical setting for the number of steps $T > 2$, so we will use this value for the remainder of our experiments. We also observe that, overall, \DPHistogram\ performs better than \GreedyHistogram\ because of the fact that it can keep a tab on the best rewrite among all the candidate rewrites it has explored for any given $T$.

\begin{figure}[t]
\epsfig{file=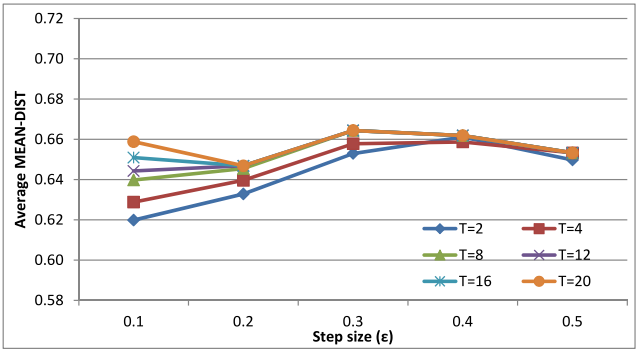, width=\columnwidth}
\caption{Effect of step size $\epsilon$ to the distance of furthest result for \GreedyHistogram.}
\label{fig:greedy-perf-over-epsilon}
\end{figure}

\begin{figure}[t]
\epsfig{file=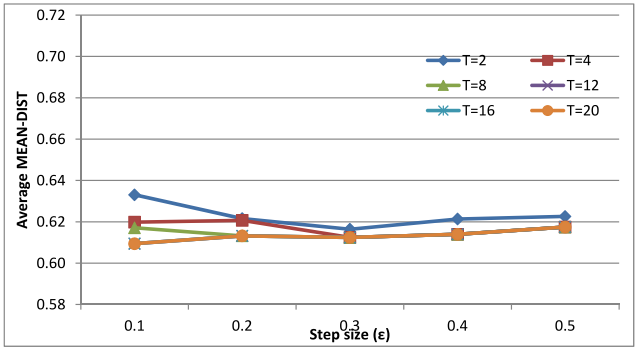, width=\columnwidth}
\caption{Effect of step size $\epsilon$ to the distance of furthest result for \DPHistogram.}
\label{fig:dp-perf-over-epsilon}
\end{figure}

\subsection{Varying Number of Steps} \label{ssec:num-steps}

We now turn to study the performance of our algorithms in terms of the amount
of steps that is allocated to them. We fixed the step size to 0.1 and look at
different step values. At a high level, we assume that, on average, each query
rewrite estimation will take approximately the same time to be computed.  To
this end, we ran all three algorithms over our data set and we compared their performance which is shown in Figure~\ref{fig:algs-perf-over-timesteps}.  In the figure, the horizontal axis is the number of steps, and the vertical axis is the average \meandist\ at a given number of steps.

\begin{figure}[t]
\epsfig{file=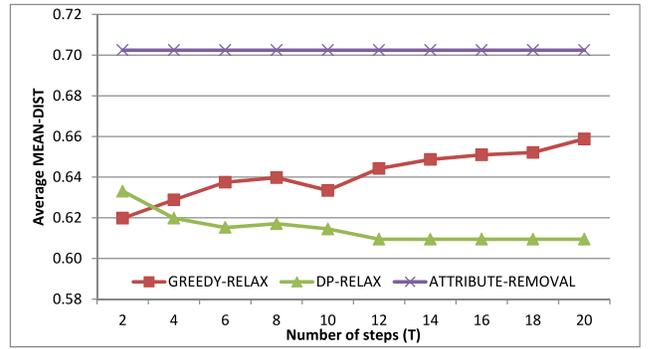, width=\columnwidth}
\caption{Average \meandist\ after a given set of time steps for all the algorithms for step size $\epsilon$$=$0.1}
\label{fig:algs-perf-over-timesteps}
\end{figure}

The first observation is that both \GreedyHistogram\ and \DPHistogram\ perform substantially better than \UnrestrictedHistogram.  The second observation is that a larger of number steps does not necessarily translate to a better performance.  This may appear counter-intuitive as one would assume that with more steps, the relaxation algorithm would get to ``explore'' the attribute space more fully to arrive at the right attribute combinations to relax.  For \GreedyHistogram, however, this needs not be the case.  This is because in cases where the estimation routine underestimates the number of results, \GreedyHistogram\ will continue to relax beyond the point necessary, leading to a set of results with higher \meandist, whereas a run with fewer number of steps will terminate with a relaxed query that it returns due to exhaustion of number of steps but lucks out in being one that retrieves sufficient number of results, leading to lower \meandist.  Indeed, the performance of \GreedyHistogram\ deteriorates after $10$ steps since the additional relaxation of the attributes only results in adding more unrelated results to the result set.

In contrast, for \DPHistogram, increasing the number of steps leads to steady improvements in \meandist.  While in principle \DPHistogram\ may be plagued by the aforementioned problem for \GreedyHistogram\ due to underestimation, because it explores the space of relaxed queries more completely, it is less affected by poor estimation compared to \GreedyHistogram. Nonetheless, by around $12$ steps, the quality of the results do not improve any further as it starts to find exactly the same relaxed query.



\begin{figure}[t]
\epsfig{file=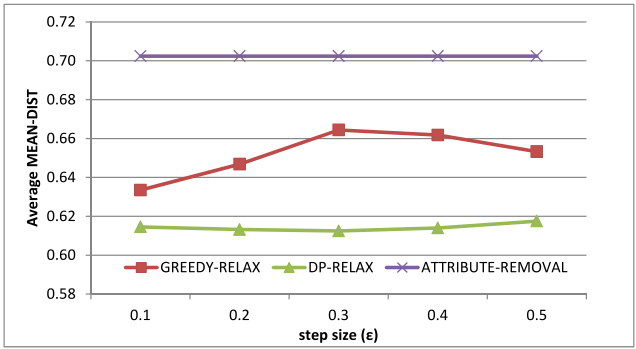, width=\columnwidth}
\caption{Average \meandist\ for different step sizes for all the algorithms for number of steps $T=10$}
\label{fig:algs-perf-over-epsilon}
\end{figure}

Finally, Figure~\ref{fig:algs-perf-over-epsilon} summarizes the relative
performance of all three algorithms for different values of the step size
$\epsilon$ fixing the number of steps $T=10$.  Again, we observe that both
\GreedyHistogram\ and \DPHistogram\ outperform \UnrestrictedHistogram.

\subsection{Testing for Attribute Dependencies}
\label{sec:ad}
As we discussed in Section~\ref{sec:relaxation}, one preprocessing step that we may apply to our algorithms is to identify attribute dependencies and drop dependent attributes from the query before rewriting it. Attribute (or functional) dependencies are very useful in optimizing queries in database systems as they can capture the relations between attributes. Our high-level intuition is that if attribute dependencies are present in our data set, it would help to take it into consideration as these dependencies point to \emph{dependence} across attributes, hence accounting for them can help with estimation, which in turn helps to find better relaxed queries. To study the presence and effect of attribute dependencies to our algorithms, we computed the conditional probabilities for all pairs of attributes and we kept only those that were higher than $0.9$.

Given a pair of attributes $a$ and $b$ where $P(a = v|b = w)$ for a large number of pairs of attribute values $v$ and $w$, we need to decide whether we should drop attribute $a$ or $b$ from the query before relaxing.  As we discussed in Section~\ref{sec:relaxation} the best choice depends on the distribution of values and the distance function for attribute $a$ or $b$. For example, if $a$ is more selective (i.e. appears in less tuples) than $b$, it may be better to drop attribute $a$ as it is expected to relax the query more than if we dropped $b$. On the other hand, dropping $b$ may also help since tuples it appears in are already partially implied by $a$.

\begin{figure}[t]
\epsfig{file=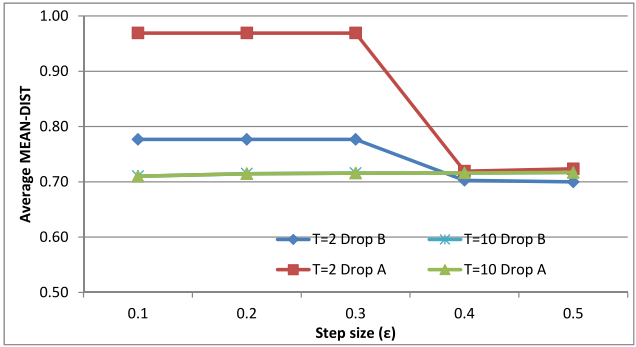, width=\columnwidth}
\caption{Average \meandist\ when for two dependent attributes $a\rightarrow b$ ($a$ implies $b$), $a$ is dropped and $b$ is dropped for different values of $\epsilon$}
\label{fig:fd-greedy-perf-over-epsilon}
\end{figure}

To this end, we repeated the experiment for identifying a good step size with the attribute-dependency preprocessing enabled. We computed the results for both alternatives for dropping an attribute (that is, either $a$ or $b$).  We report the results in Figure~\ref{fig:fd-greedy-perf-over-epsilon} for the \GreedyHistogram\ algorithm using small and large values of $T$ ($T=2$ and $T=10$) respectively.  The findings for the \DPHistogram\ algorithm are similar.


The overall result is surprising, as we find that either approach of incorporating attribute dependencies by dropping attribute $a$ or $b$ have not led to better performance, and in some cases even worse performances.  To understand this better, we perform a query-by-query analysis of the results, and found that the problem manifests itself due to a complex chain of interactions.  First, a significant fraction of these queries are \WebQuery{<brand> <model> query}. The attribute dependencies we found are also between attribute \emph{brand} and \emph{model}, where each model is associated with a unique brand.  As mentioned in Section~\ref{sssec:distance-function}, we do not have many distances estimated between models due to data sparsity.  When the attribute \emph{model} is dropped, we retrieve a number of different models of the same brand, all of which are considered to be quite far away from the query as we treat missing distances as $1$.  When the attribute \emph{brand} is dropped, the situation is even worse as the algorithm will now relax the attribute \emph{model} to close to distance $1$ in order to find sufficient number of results due to missing distances.  Hence, in such cases, performances are worse than not dropping attribute at all, as the results are now no longer constrained by \emph{brand}.



\subsection{Index Performance}
In another experiment, we measured the work done by the index in terms of the
number of documents processed by the index.  The processing done by the index
typically includes computing ranking features and scoring the document for the
given query.  As the processing takes time,  one would like the number of
documents processed by the index close to the documents estimated by the
rewrite algorithm.  Figure~\ref{fig:new-algs-numhits-over-numsteps} illustrates
the performance of the algorithms in terms of processing done by the index for
step size $\epsilon$$=$0.1.

\begin{figure}[t]
\epsfig{file=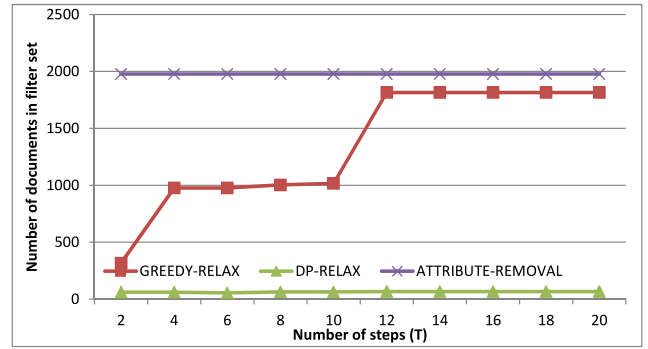, width=\columnwidth}
\caption{The median of the number of results processed by the index using all three algorithms for different values of $T$ and for step size $\epsilon$$=$0.1}
\label{fig:new-algs-numhits-over-numsteps}
\end{figure}

The general trend is that \DPHistogram\ produces rewrites that give close to the desired number of results of $k = 10$, and generates the least work for the index among the three algorithms.  On the other extreme, \UnrestrictedHistogram\ produces rewrites that generate the most work for the index due to its choice of removing the chosen attribute completely.  \GreedyHistogram\ spans the performance gap between these two algorithms. For lower values of $T$, it results in smaller number of documents to be included in the filter set while at the higher values of $T$, it comes close to \UnrestrictedHistogram\ in terms of the number of documents admitted into the filter set. The reason for the behavior exhibited by \GreedyHistogram\ is as follows.  As \GreedyHistogram\ explores one attribute at a time, and chooses its next step based on its current relaxed query, it performs a depth-first-like search through the space of relaxed queries.  In many cases, due to its choice in prioritizing the relaxation in favor of the most selective attribute, it ends up repeatedly relaxing the same attribute leading to completely relaxing an attribute.  These type of relaxed queries typically leads to retrieving significantly more number of results.  Note however that the result set may still have similar average quality as measured by \meandist, as confirmed by the figures in the previous sections.

\section{Conclusion}
\label{sec:conclusion}

In this paper we propose a query-rewrite framework for answering structured web queries when users pose queries that would have led to very few results.  Our framework takes into account the stringent time requirement of answering web queries, and balances it with the need of retrieving results close to the user queries.  We describe two approaches to solving this problem, and show experimentally that both solutions produce meaningful results given our constraints.

After studying the performance of the three algorithms with respect to parameters like step size and the number of rewrites to explore, we conclude that if time envelope admits more rewrites, then \DPHistogram\ is more applicable.  In the case of extremely small latency restrictions, \GreedyHistogram\ is a better choice.

The approaches proposed in this paper is especially important in domains where there is an underlying source of structured data, but for which users lacking domain expertise may end up issuing queries that have few or even zero matches.  This contributes to the growing literature on how to efficiently surface structured results in response to web queries.



\end{document}